\documentclass[11pt]{article}

\usepackage[in]{fullpage}

\usepackage{aliascnt}
\usepackage{amsmath,amssymb,amsfonts,amsthm}
\usepackage{graphicx}
\usepackage{stmaryrd}
\usepackage{hyperref}
\usepackage{xspace}

\hypersetup{%
   breaklinks,%
   ocgcolorlinks, colorlinks=true,%
   urlcolor=[rgb]{0.25,0.0,0.0},%
   linkcolor=[rgb]{0.5,0.0,0.0},%
   citecolor=[rgb]{0,0.2,0.445},%
   filecolor=[rgb]{0,0,0.4}, anchorcolor=[rgb]={0.0,0.1,0.2}%
}

\newcommand{\ACSF}{ACSF\xspace}%
\newcommand{\CSF}{CSF\xspace}%

\newtheorem{theorem}{Theorem}

\newaliascnt{lemma}{theorem} \newtheorem{lemma}[lemma]{Lemma}
\aliascntresetthe{lemma} 

\newaliascnt{corollary}{theorem}
\newtheorem{corollary}[corollary]{Corollary}
\aliascntresetthe{corollary}

\newaliascnt{observation}{theorem}
\newtheorem{observation}[observation]{Observation}
\aliascntresetthe{observation}

\newaliascnt{conjecture}{theorem}
\newtheorem{conjecture}[conjecture]{Conjecture}
\aliascntresetthe{conjecture}

\newaliascnt{claim}{theorem}

\aliascntresetthe{claim}

\theoremstyle{remark}
\newaliascnt{remark}{theorem}

\aliascntresetthe{remark}

\usepackage{todonotes}

\newcommand{\N}{\mathbb N}
\newcommand{\Z}{\mathbb Z}
\newcommand{\R}{\mathbb R}
\DeclareMathOperator{\CH}{\mathcal{CH}}

\newcommand{\xmax}{x_{\mathrm{max}}}
\newcommand{\ymax}{y_{\mathrm{max}}}
\newcommand{\eps}{\varepsilon}

\newcommand{\Lidicky}{Lidick\'y\xspace}
\providecommand{\Barany}{B{\'a}r{\'a}ny\xspace}
\providecommand{\Jarnik}{Jarn\'ik\xspace}
\providecommand{\Mobius}{M\"obius\xspace}
\newcommand{\si}[1]{#1}

\newcommand{\etal}{\textit{et~al.}\xspace}%
\renewcommand{\th}{th\xspace}%

\title{Grid peeling and the affine curve-shortening flow\footnote{A preliminary version of this paper appeared in \emph{Proceedings of the Twentieth Workshop on Algorithm Engineering and Experiments (ALENEX 2018)}, pp. 109--116, SIAM, 2018.}}

\author{%
   David Eppstein%
   \thanks{Department of Computer Science, University of California,
      Irvine. Supported in part by the National Science Foundation
      under Grants CCF-1228639, CCF-1618301, and CCF-1616248.}%
   \and%
   Sariel Har-Peled%
   \thanks{Department of Computer Science; University of Illinois; 201
      N. Goodwin Avenue; Urbana, IL, 61801, USA; {\tt
         sariel@illinois.edu}; {\tt \url{http://sarielhp.org/}.} Work
      on this paper was partially supported by a NSF AF awards
      CCF-1421231 and CCF-1217462.}%
   \and%
   Gabriel Nivasch%
   \thanks{Corresponding author. Department of Computer Science, Ariel University, Ariel,
      Israel. {\tt gabrieln@ariel.ac.il}.}}%
\date{}

\begin{document}
\maketitle

\begin{abstract}
    In this paper we study an experimentally-observed connection
    between two seemingly unrelated processes, one from computational
    geometry and the other from differential geometry. The first one
    (which we call \emph{grid peeling}) is the convex-layer
    decomposition of subsets $G\subset \Z^2$ of the integer grid,
    previously studied for the particular case $G=\{1,\ldots,m\}^2$ by
    Har-Peled and \Lidicky (2013). The second one is the
    affine curve-shortening flow (\ACSF), first studied by Alvarez
    \etal (1993) and Sapiro and Tannenbaum (1993). We
    present empirical evidence that, in a certain well-defined sense,
    grid peeling behaves at the limit like \ACSF on convex curves. We offer some
    theoretical arguments in favor of this conjecture.

    We also pay closer attention to the simple case where $G=\N^2$ is
    a quarter-infinite grid. This case corresponds to \ACSF starting
    with an infinite L-shaped curve, which when transformed using the
    \ACSF becomes a hyperbola for all times $t>0$. We prove that, in
    the grid peeling of $\N^2$, (1) the number of grid points removed
    up to iteration $n$ is $\Theta(n^{3/2}\log n)$; and (2) the
    boundary at iteration $n$ is sandwiched between two hyperbolas
    that are separated from each other by a constant factor.
\end{abstract}

\section{Introduction}

Let $G$ be a planar point set. The \emph{convex-layer decomposition}
(or \emph{onion decomposition}) of
$G$~\cite{b-omd-76,c-clps-85,d-co-04,e-chp-82,hl-pg-13} is a discrete
algorithmic process in which points of $G$ are iteratively removed, as
follows: Let $G_0=G$. Then, for each $n\ge 1$ such that
$G_{n-1}\neq \emptyset$, let $H_n = \CH{(G_{n-1})}$ (the convex hull
of the current set), let $L_n$ be the set of vertices of $H_n$, and
remove $L_n$ from the current set by setting
$G_n = G_{n-1} \setminus L_n$.\footnote{Note that $G_{n-1}$ might contain points which lie on the boundary of $H_n$ but are not vertices. These points will still be present in $G_n$.} We call $H_n$ the \emph{$n$\th convex
   layer} of $G$.  This decomposition has applications in
range-searching data structures~\cite{cgl-pgd-85} and as a measure of
depth in robust statistics~\cite{b-omd-76,e-chp-82}.

Motivated by the question of whether grid points behave similarly to
random points, Har-Peled and \Lidicky \cite{hl-pg-13} studied the
convex-layer decomposition of the $m\times m$ integer grid
$G=\{1,\ldots,m\}^2$. They proved that this point set has
$\Theta(m^{4/3})$ convex layers. They also briefly noted that the
convex layers of this point set appear to converge to circles as the
process advances.

In this paper we explore an experimentally-observed connection between
the convex-layer decomposition of more general subsets $G\subset \Z^2$
of the integer grid (which we call \emph{grid peeling}), and a
continuous process on smooth curves known as the \emph{affine
   curve-shortening flow} (\ACSF). Our conjectural connection between
these two processes, if true, would show that in the square case
studied by Har-Peled and \Lidicky, peeling indeed converges to a
circular shape. More generally, it would show that for any convex
shape, in the limit as the grid density becomes arbitrarily fine, the
result of peeling the intersection of that shape with a grid converges
to an ellipse.

\subsection{The affine curve-shortening flow}

In the affine curve-shortening flow, a smooth curve
$\gamma\subset \R^2$ varies with time in the following way. At each
moment in time, each point of $\gamma$ moves perpendicularly to the
curve, towards its local center of curvature, with instantaneous
velocity $r^{-1/3}$, where $r$ is that point's radius of curvature at
that time. Thus, for a smooth convex curve, all points move inwards,
possibly at different velocities. For non-convex curves, points of
local non-convexity move outwards. See \autoref{fig:A:C:S:F}.

\begin{figure}
    \centerline{\includegraphics[width=6cm]{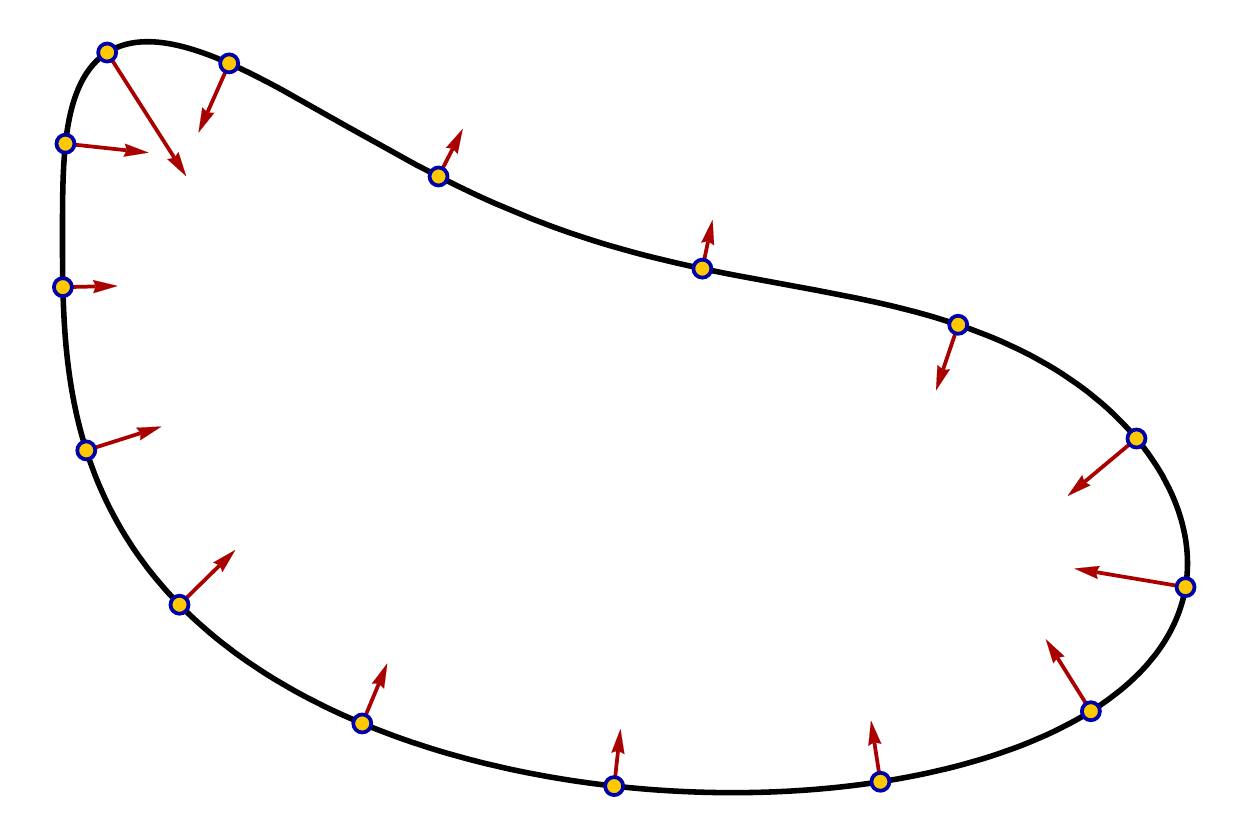}}
    \caption{%
       Affine curve-shortening flow. The arrows
       indicate the instantaneous velocity of different points along
       the curve at the shown time moment.}
    \label{fig:A:C:S:F}%
\end{figure}

The \ACSF was first studied by Alvarez \etal \cite{aglm-afeip-93} and
Sapiro and Tannenbaum~\cite{st-aiss-93}.  It differs from the more
usual \emph{curve-shortening flow} (\CSF)~\cite{c-gceip-03,cz-csp-01},
in which each point moves with instantaneous velocity $r^{-1}$.
Unlike the \CSF, the \ACSF is invariant under affine transformations:
Applying an affine transformation to a curve, and then performing the
\ACSF, gives the same results (after rescaling the time parameter
appropriately) as performing the \ACSF and then applying the affine
transformation to the shortened curves. Moreover, if the affine
transformation preserves area, then the time scale is unaffected.  For
more on the \ACSF see~\cite{c-gceip-03,c-cccas-15,i-cccas-16} and
references cited there.

For the \CSF, every smooth Jordan curve eventually becomes convex and
then converges to a circle as it collapses to a point, without ever
crossing itself.  Angenent \etal~\cite{ast-ahenc-98} proved that,
correspondingly, under the \ACSF, every smooth Jordan curve becomes
convex and then converges to an ellipse as it collapses to a point,
without self-crossings.

Even if the initial curve $\gamma$ is not smooth (e.g.~it has sharp corners), as long as it satisfies certain natural conditions, there exists a unique time-dependent curve $\gamma(t)$ which satisfies the ACSF (or the CSF) condition for all $t>0$, and which converges to $\gamma$ as $t\to 0^+$. See \cite{c-gceip-03}, Theorems 3.26 and 3.28.

The \ACSF was originally applied in computer vision, as a way of
smoothing object boundaries~\cite{c-gceip-03} and of computing shape
descriptors that are insensitive to the distortions caused by changes
of viewpoint. Because peeling can be computed quickly and efficiently,
by a purely combinatorial algorithm~\cite{c-clps-85}, our conjectural
connection between peeling and the \ACSF could potentially provide an
efficient way of performing these computations.  However, to fully
realize this potential application, it would be helpful to prove
rigorous bounds on the accuracy of approximation, and to find a way to
generalize the approximation so that it can handle non-convex curves
as well. In the other direction, our conjecture would allow us to
apply results on the well-understood behavior of the \ACSF to the less
well-understood algorithmic process of grid peeling. For instance, it
would explain the circular layer shapes observed by Har-Peled and
\Lidicky.

\subsection{Organization of this paper}

This paper is organized as follows.  In \autoref{the_connection} we
formalize our conjectured connection between peeling and the \ACSF as
\autoref{conj:A:C:S:F}, and provide a non-rigorous justification for
the conjecture.  In \autoref{app} we describe our implementation
details, and report on more detailed experiments that quantify the
similarity between peeling and the \ACSF.  In \autoref{sec_bounded_reg}
we prove \autoref{theorem_4over3}, which shows that for bounded regions, the rates of
peeling and the \ACSF are within a constant factor of each other, a
weaker form of our conjecture.  In \autoref{sec_peel_corner} we
examine more closely a special case of our conjecture on a
quarter-infinite grid, and prove more precise results for that case.

\section{The connection}
\label{the_connection}

Empirical evidence points to a connection between grid peeling and the
\ACSF.  For a curve $\gamma$, let $\gamma(t)$, $t\ge 0$, be the result
of applying \ACSF on $\gamma$ for time duration $t$. Given a positive
integer $n$, let $(\Z/n)^2$ be the uniform grid with spacing
$1/n$. Given a convex region $R\subset\R^2$, let
$G_{[n]}(R) = R\cap (\Z/n)^2$ be the set of grid points of $(\Z/n)^2$
contained in $R$. Informally, for a convex curve $\gamma$, we have
that peeling $G_{[n]}(\CH(\gamma))$ appears to approximate the \ACSF
on $\gamma$ as $n\to\infty$.

This connection is illustrated in
\autoref{fig_round_tr}. \autoref{fig_round_tr} (left) shows the \ACSF
evolution of a sample convex curve $\gamma$, given by
$\gamma=\{(x(a), y(a)) : 0\le a < 2\pi\}$ for
$x(a) = ((1-\sin{a})/2)^2$ and $y(a) =
((1-\sin{(a+2)})/2)^{1.3}$. Specifically, the figure shows
$\gamma(0.02t)$ for $t=0,1,2, \ldots,14$.  \autoref{fig_round_tr}
(center) shows every fifth layer of the convex-layer decomposition of
$G_{[30]}(R)$ for $R=\CH(\gamma)$. The similarity to
\autoref{fig_round_tr} (left) is immediately evident. Finally,
\autoref{fig_round_tr} (right) shows every $2714$\th layer of the
convex-layer decomposition of $G_{[5000]}(R)$. \autoref{fig_round_tr}
(left) and (right) are virtually indistinguishable to the naked
eye.

\begin{figure}
    \centerline{\includegraphics{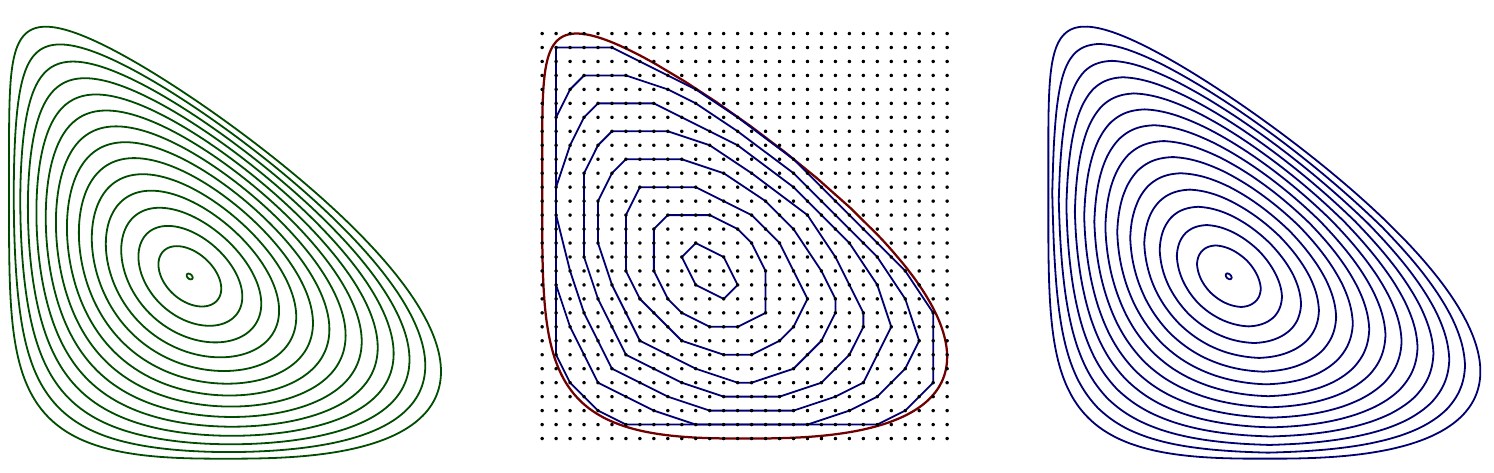}}
    \caption{\label{fig_round_tr}Left: \ACSF evolution of a convex
       curve. Center and right: Convex-layer decomposition of the set
       of grid points inside the same convex curve, for different grid
       spacings.}
\end{figure}

We can formalize this resemblance by the following conjecture.

\begin{conjecture}%
    \label{conj:A:C:S:F}%
    There exists a constant $c\approx 1.6$ such that the following is
    true: Let $R\subset \R^2$ be a convex region, and let
    $\gamma=\partial R$ be its boundary. Let $t^*$ be the time it takes for $\gamma$
    to collapse to a point under the \ACSF (or $t^*=\infty$ for
    unbounded sets that never collapse). Fix a time $0\le t<t^*$, and
    let $\gamma' = \gamma(t)$ under the \ACSF. For a fixed $n$, let
    $G'$ be the $m$\th convex layer of $G_{[n]}(R)$ for
\begin{equation}\label{eq_conj}
m = c t n^{4/3}.
\end{equation}
Then, as $n\to\infty$, the boundary of the
    convex hull of $G'$ converges pointwise to $\gamma'$.
\end{conjecture}

In particular, the \ACSF is known to converge to an ellipse for any
closed initial boundary $\gamma$, in the limit as $t\to t^*$, when its shape
is rescaled to have constant area.  Correspondingly, by the
conjecture, the convex layers of $G_{[n]}(R)$ should also converge to
ellipses as $t\to t^*$ and $n\to\infty$. By symmetry, the convex
layers of a square grid should indeed converge to circles.

\subsection{Justification for \autoref{conj:A:C:S:F}}

One intuitive but somewhat vague justification for
\autoref{conj:A:C:S:F} is that the \ACSF is invariant under affine
transformations (in fact, it is the unique affine-invariant flow of
least order~\cite{c-gceip-03}), and grid peeling is also invariant
under a subgroup of affine transformations, namely the ones that
preserve the unit grid.

A more detailed justification is as follows. Balog and
\Barany~\cite{bb-ochip-91} proved that, if $R$ is the unit disk, then
$\CH(G_{[n]}(R))$ has $\Theta(n^{2/3})$ vertices. Equivalently, if $R$
is a disk of radius $r$, then $C=\CH(G_{[n]}(R))$ has
$\Theta((nr)^{2/3})$ vertices. Let us assume these vertices are
uniformly distributed along the boundary of $C$,\footnote{This seems
   to be the case empirically.} so a portion of $\partial C$ of length
$d$ contains $\Theta(dn^{2/3}r^{-1/3})$ vertices.

Now, let $R\subset \R^2$ be an arbitrary convex region with smooth
boundary $\gamma=\partial R$, and fix a small portion $\delta$ of
$\gamma$, of almost constant radius of curvature $r$. Let $d$ be the
length of $\delta$. Let $C=\CH(G_{[n]}(R))$ for large $n$. Then the
portion $\delta'$ of $\partial C$ that is close to $\delta$ contains
$\Theta(dn^{2/3}r^{-1/3})$ vertices. Let $\eps>0$ be much smaller than
$d$. In order for $\delta'$ to advance inwards by distance $\eps$, a
total of $\eps dn^2$ grid points must be removed. This should take
$\Theta(\eps n^{4/3}r^{1/3})$ iterations. Therefore, $\delta'$ should
move inwards at speed $\Theta(n^{-4/3}r^{-1/3})$. This is
$\Theta(n^{4/3})$ times slower than \ACSF, independently of~$r$.

\section{Implementation and experiments}\label{app}

We first implemented a simple front-tracking \ACSF approximation method
that works as follows. We sample a number $m$ of points
$p_1, \ldots, p_m$ along the given curve $\gamma$. For each point
$p_i$, we estimate the normal vector and the radius of curvature at
$p_i$ by the normal vector $v_i$ and radius $r_i$ of the unique circle
passing through points $p_{i-1},p_i, p_{i+1}$. We simultaneously let
all points move at the appropriate speeds for a short time interval
$t = c\cdot (d_{\mathrm{min}})^{4/3}$, where $d_{\mathrm{min}}$ is the minimum
distance between two consecutive points, and $c$ is a fixed
parameter. Then we repeat the process. Hence, as the sample points get
closer and closer, we take smaller and smaller time steps. Here the exponent $4/3$ was chosen in order for the simulation to be scale-independent.

A disadvantage of this method is that, as the curve becomes
elliptical, the sample points tend to bunch together at the sharp ends
of the ellipse, causing the time step to decrease very drastically.
In order to overcome this problem, we then implemented a more
sophisticated approach, in which each point is also given a tangential
velocity component $w_i$ (i.e. $w_i\perp v_i$). (Tangential velocities
should not affect the evolution of a flow, since they only cause curve
points to move within the curve.) We make the length of $w_i$
proportional to $\|v_i\|\cdot\log
(\|p_i-p_{i-1}\|/\|p_i-p_{i+1}\|)$. Hence, if $p_i$ is equidistant
from $p_{i-1}$ and $p_{i+1}$, then $\|w_i\|=0$. Otherwise, if $p_i$ is
closer to $p_{i-1}$ than to $p_{i+1}$, say, then $w_i$ points in the
direction of $p_{i+1}$.

This simple approach was enough for our purposes. For more advanced
flow simulation methods, see e.g.~\cite{c-gceip-03,ef-acsfm-17,moisan} and references cited there.

Our \ACSF C++ program may be found at \textsf{\si{ACSF.cpp}}, in the
ancillary files of this paper.

For the grid peeling simulations, we represent the grid subset as a
one-dimensional array that stores, for each row, the $x$-coordinates
of the leftmost and rightmost grid points in that row. We compute the
convex hull at each iteration using Andrew's modification of Graham's
scan~\cite{a-aeach-79,bcko-cgaa-08}. Thus, to find the
$\Theta(n^{4/3})$ layers of an $n\times n$ grid we take $O(n)$ time
per layer and $O(n^{7/3})$ time overall. Faster $O(n^2\log n)$-time
algorithms are possible~\cite{c-clps-85,nielsen,rr} but were unnecessary for our
experiments.  We implemented this peeling algorithm in two C++
programs, ``\textsf{\si{peel N2.cpp}}'' (for peeling $\N^2$) and
``\textsf{\si{peel shape.cpp}}'' (for peeling general shapes).

\subsection{Experiments on bounded regions}

In order to test \autoref{conj:A:C:S:F}, we ran both ACSF and grid peeling on several bounded convex regions, and compared the results. The regions we used are: $R_1 = \CH(\gamma)$ for the curve $\gamma$ of \autoref{fig_round_tr}; $R_2$, a square of side $1$; $R_3$, a triangle with vertices $(0,0)$, $(1,3/4)$, $(2/5,1)$; $R_4$, a half-disk of diameter $1$; and $R_5$, a disk of diameter $1$.

\begin{figure}
\centerline{\includegraphics{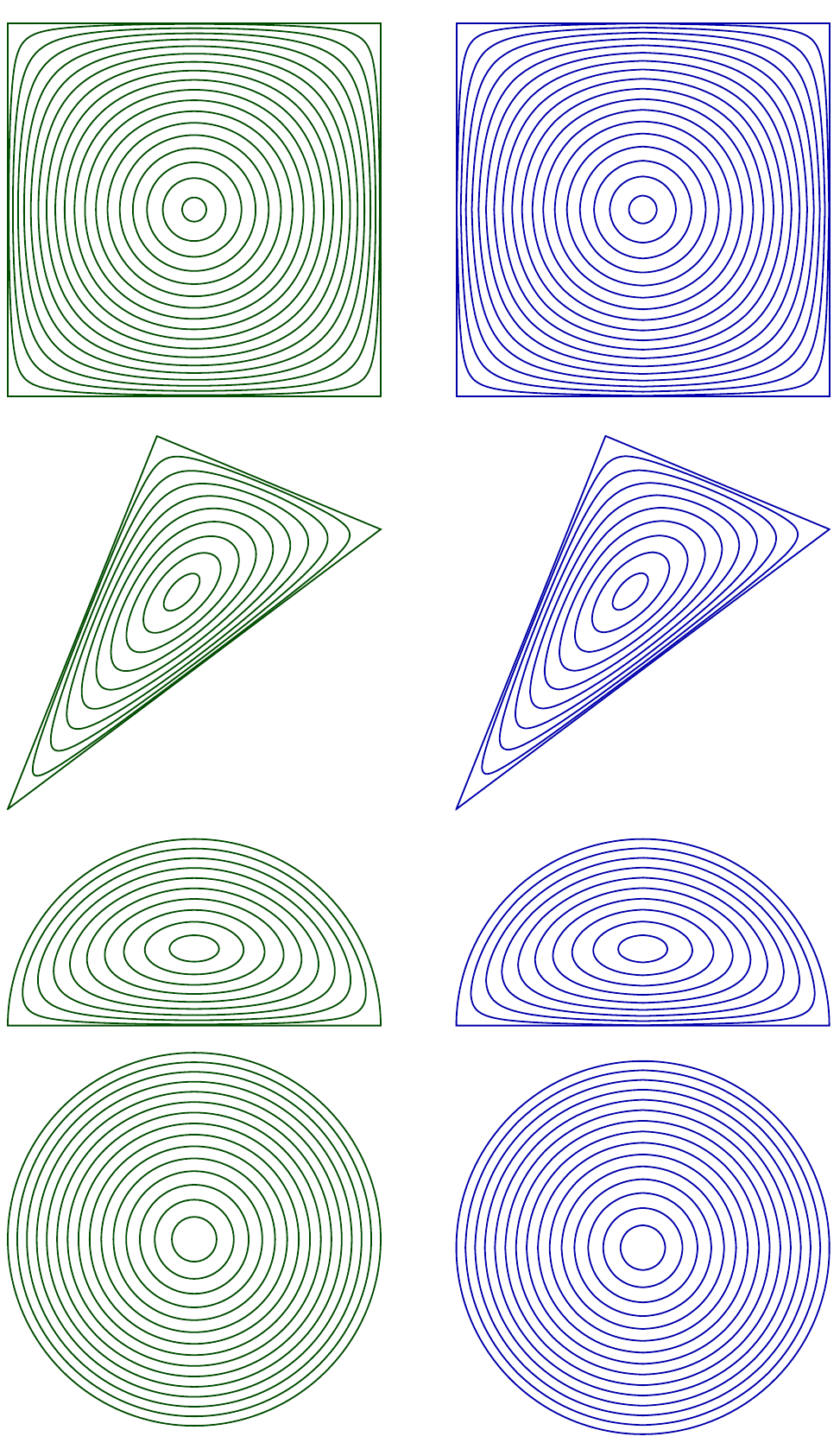}}
\caption{\label{fig_compare_shapes}Comparison between ACSF (left) and grid peeling (right) on several test shapes.}
\end{figure}

\autoref{fig_compare_shapes} shows the results for $R_2, \ldots, R_5$. Just as in \autoref{fig_round_tr}, here each left figure shows \ACSF with time steps of $0.02$, and each right figure shows every $2714$th convex layer, starting with a grid of spacing $1/5000$. Each left figure is barely distinguishable to the naked eye from the corresponding right figure.

In order to further test \autoref{conj:A:C:S:F}, we took the same regions $R_1, \ldots, R_5$, and measured the Hausdorff
distance between the results of the two processes, for increasing
values of the grid density $n$. For each $R_i$, we first ran our \ACSF simulation with higher-precision parameters, until the times $t_1, \ldots, t_5$ at which
the area enclosed by the curve decreased to $95\%, 90\%, \ldots, 75\%$ of its original area.\footnote{Actually, for $R_5$ we did not simulate \ACSF. We simply used the closed-form solution given by $r(t) = (r(0)^{4/3} - 4t/3)^{3/4}$, where $r(t)$ is the radius of the circle at time $t$.} Then we ran grid peeling using a variety of grid spacings; specifically, $1/n$ for $n=1000, 3000, 10000, 30000, 100000$. In each case, we ran the process until the times $m_1, \ldots, m_5$ at which the number of grid points decreased to $95\%, \ldots, 75\%$ of its original value.

For each case, we then computed the Hausdorff distance between the
\ACSF curve and the grid-peeling curve, both represented as polygonal
chains.\footnote{We computed the Hausdorff distances using a simple brute-force approach, using the fact that for convex polygonal chains the maximum distance is attained by a vertex (Atallah~\cite{atallah}). (Atallah~\cite{atallah} also presents a more efficent Hausdorff-distance algorithm for convex polygonal chains, which we did not use.)} (For comparison, we also computed the initial Hausdorff distance, which reflects the inherent inaccuracy in approximating the given smooth curve by grid points.) \autoref{fig_hausdorff} shows the results.

\begin{figure}
\centerline{\includegraphics{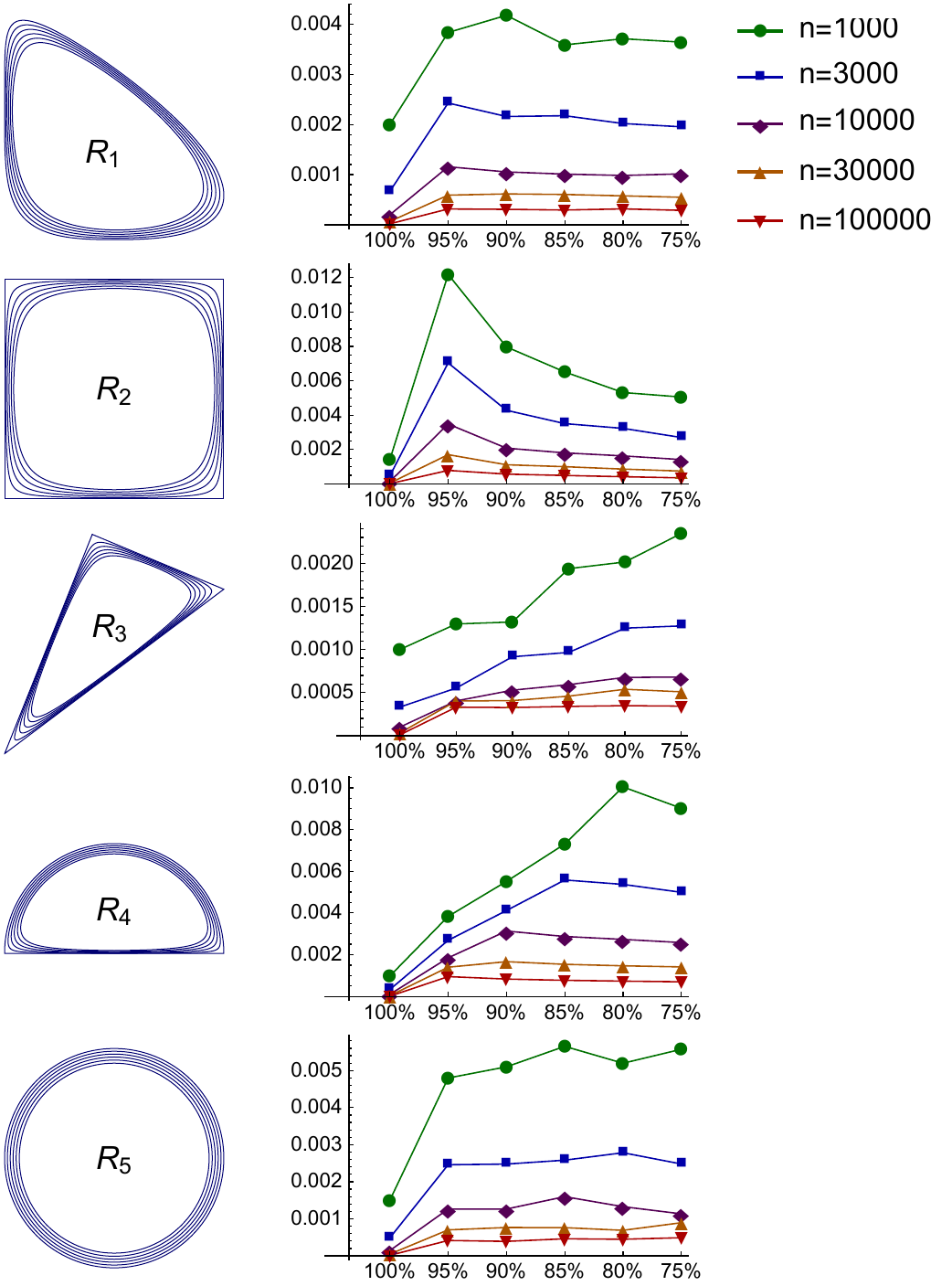}}
\caption{\label{fig_hausdorff}Hausdorff distance between the results of \ACSF and grid peeling, for a variety of test shapes, times, and grid spacings.}
\end{figure}

As we can see, the Hausdorff distance decreases with increasing $n$. Furthermore, for large values of $n$, the length of time has no major effect on the Hausdorff distance.

Finally, we checked whether the \ACSF times $t_1, \ldots, t_5$ are related to the grid peeling times $m_1, \ldots, m_5$ in the manner predicted by \autoref{conj:A:C:S:F}. To do this, we solved for the constant $c$ in (\ref{eq_conj}), and computed the approximations $c \approx m_i/(t_in^{4/3})$. The results are shown in \autoref{fig_capprox}. As can be seen, in all cases we obtained values close to $1.6$; furthermore, the approximations get closer to each other as either $n$ or $t$ increases.

\begin{figure}
\centerline{\includegraphics{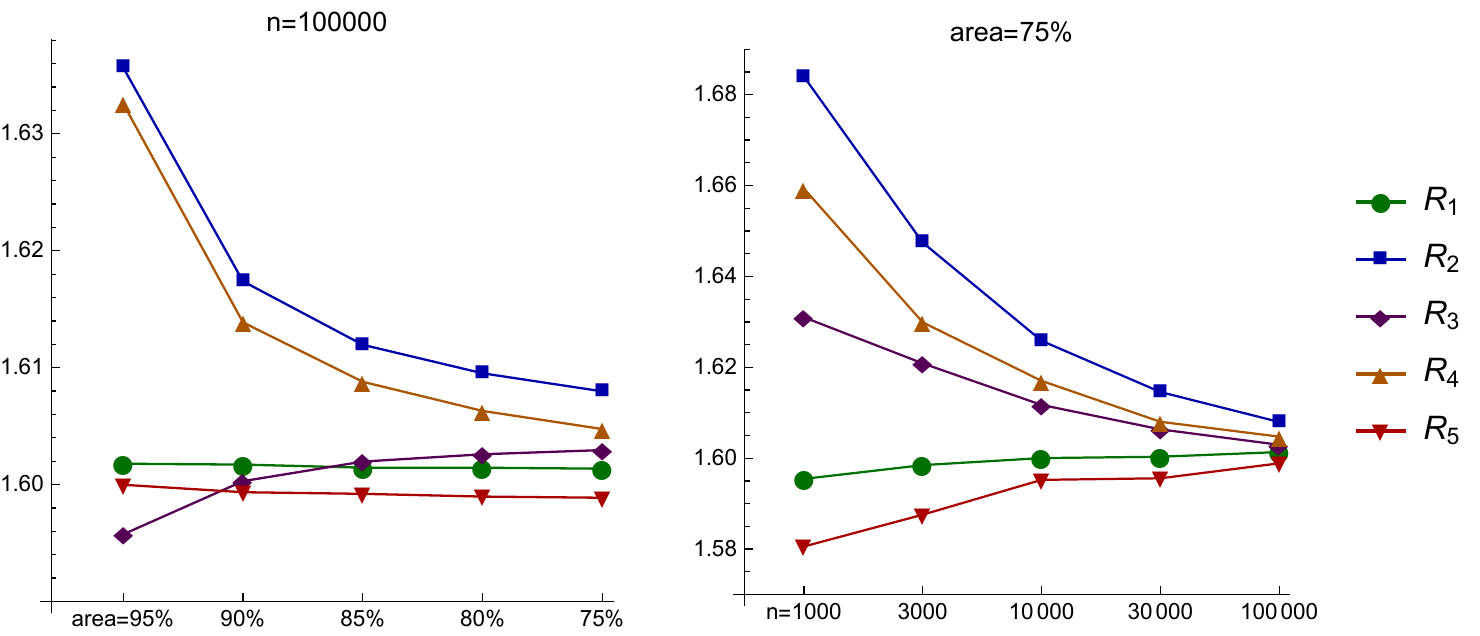}}
\caption{\label{fig_capprox}Approximations of the constant $c$ given by the experiments. The left plot uses the data for $n=100000$ and various values of the final area. The right plot uses the data for a final area of $75\%$ and various values of $n$.}
\end{figure}

\section{Number of convex layers for bounded regions}
\label{sec_bounded_reg}

In this section we prove that for bounded regions $R$, \autoref{conj:A:C:S:F} is asymptotically correct as far
as the number of convex layers is concerned:

\begin{theorem}\label{theorem_4over3}
    Let $R\subset \R^2$ be a bounded convex region. Then the number of
    convex layers of $G_{[n]}(R)$ is $\Theta(n^{4/3})$, with a
    constant of proportionality that depends on $R$.
\end{theorem}

Correspondingly, in \ACSF, if the given curve $\gamma$ is dilated by a
factor of $k$, then its evolution is dilated in time by a factor of
$k^{4/3}$.

\autoref{theorem_4over3} follows from the result of~\cite{hl-pg-13} on
the number of convex layers of square grids. First, we note that the
result of~\cite{hl-pg-13} can be readily generalized to rectangular
grids using the same argument:\footnote{The techniques
   of~\cite{hl-pg-13} are also presented in \autoref{sec_peel_corner}
   below.}

\begin{lemma}%
    \label{lemma_H:L:rectangle}%
    Let $m, n$ be integers satisfying $\sqrt{m}\le n\le m^2$. Then the
    number of convex layers of $G=\{1,\ldots,m\}\times\{1,\ldots,n\}$
    is $\Theta\bigl(m^{2/3}n^{2/3}\bigr)$.
\end{lemma}

Now, let $R\subset \R^2$ be a given bounded convex region. By John's
ellipsoid theorem~\cite{j-epis-48}, there exist two ellipses that
satisfy $E_1 \subseteq R \subseteq E_2$, such that the ratio between
their areas is at most $4$.

Let $n$ be an integer, and let $G=G_{[n]}(R)$. Scale up all these sets
by a factor of $n$, obtaining $R'$, $E'_1$, $E'_2$, and
$G'\subset \Z^2$. The grid peeling process is clearly invariant to
linear transformations. We now focus on \emph{grid-preserving} linear
transformations, that is, linear transformations that map $\Z^2$
bijectively into $\Z^2$. A linear transformation $f:\R^2\to\R^2$ is
grid-preserving if and only if it is of the form $f(p)=Mp$, where $M$
is a $2\times 2$ integer matrix with determinant $\pm 1$.

\begin{lemma}%
    \label{lem_grid-pres}%
    Let $v_1,v_2\in\Z^2$ be linearly independent vectors. Then there
    exists a grid-preserving linear transformation $f$ that maps $v_1$
    into the $x$-axis, and such that $f(v_2)$ has slope at least $2$
    in absolute value.
\end{lemma}

\begin{proof}
    We first apply a grid-preserving linear transformation $f_1$ that
    maps $v_1$ to the $x$-axis. Denote $f_1(v_2)=(x_2,y_2)$. Then we
    apply a horizontal shear $f_2:(x,y)\mapsto(x-my,y)$, for an
    appropriate $m\in\Z$. The appropriate $m$ is either
    $\lfloor y_2/x_2 \rfloor$ or $\lceil y_2/x_2 \rceil$.
\end{proof}

Now, the ellipse $E'_1$ contains a rectangle $T_1$ whose area is at
least a constant fraction of the area of $E_1$. Applying
\autoref{lem_grid-pres}, we turn $T_1$ into a parallelogram $T'_1$
with two horizontal sides and two shorter, close-to-vertical
sides. Hence, $T'_1$ contains an axis-parallel rectangle $T''_1$ whose
area is at least a constant fraction of the area of $T'_1$. If $n$ is
large enough, then the side lengths $m_1$ and $m_2$ of $T''_1$ will
satisfy $\sqrt m_1 \le m_2\le m_1^2$. Therefore, we can apply the
lower bound of \autoref{lemma_H:L:rectangle} on $T''_1$.

The upper bound proceeds similarly, using ellipse $E'_2$. This completes the proof of \autoref{theorem_4over3}.

\section{Peeling a quarter-infinite grid}\label{sec_peel_corner}

A simple test case for \autoref{conj:A:C:S:F} is the region
$R=\{(x,y) : x\ge 0, y\ge 0\}$ (the first quadrant of the plane). In
this case, the grid spacing $n$ is irrelevant, so we can simply take
$G=\N^2$ (where $\N = \{0,1,2,\ldots\}$). The boundary of $R$ is the
L-shaped curve $\gamma_L=\{(x,0):x\ge 0\}\cup\{(0,y):y\ge 0\}$.

The time-dependent hyperbola
\begin{equation}\label{eq_gammaL}
    \gamma_L(t) = \{(x,y) : y=(4/3^{3/2}) t^{3/2}/x\}.
\end{equation}
satisfies the \ACSF condition for all $t>0$ (as can be verified by a
simple calculation), and it converges to $\gamma_L$ as $t\to 0^+$. The
hyperbola (\ref{eq_gammaL}) is the only solution satisfying this
property.\footnote{The existence of a unique solution was proven for
   doubly-differentiable curves, without the assumption of closedness,
   by Angenent \etal~\cite[Section 6]{ast-ahenc-98} and stated for
   closed curves without the assumption of smoothness
   in~\cite[Theorem~3.28]{c-gceip-03}. In the case here, the
   uniqueness of the solution can be proven by applying the result for
   closed curves to the boundary of a large square: If the
   quarterplane had multiple solutions, they could be approximated
   arbitrarily well by the solution near the corner of a large enough
   square, which would necessarily also have multiple solutions,
   violating \cite[Theorem~3.28]{c-gceip-03}.}
	
Hence, by \autoref{conj:A:C:S:F}, we would expect the convex layers of $\N^2$ to approach hyperbolas as the process goes on. Indeed, this is what occurs experimentally. In the next subsections we present our experimental and theoretical results regarding the convex layers of $\N^2$. But let us first introduce some notation.

\begin{figure}
    \centerline{\includegraphics{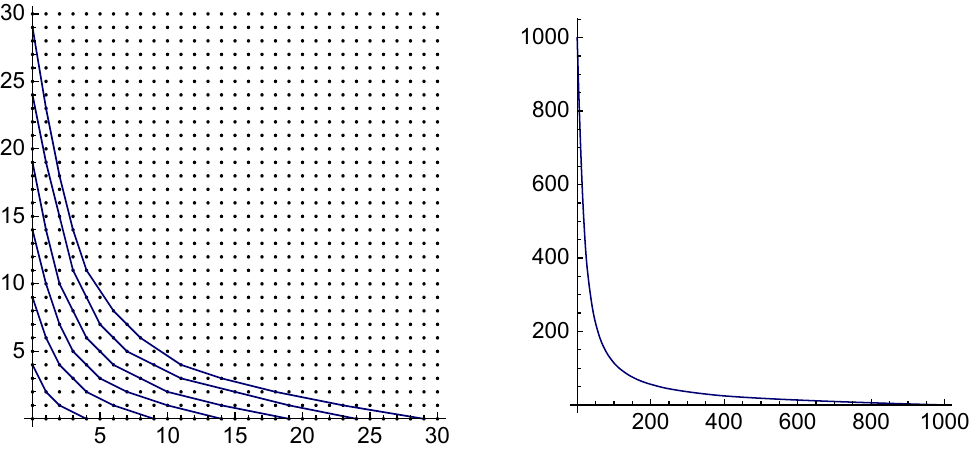}}
    \caption{\label{fig_N2layers}Left: Convex layers
       $5, 10, 15, \ldots, 30$ of $\N^2$. Right: $1000$\th convex
       layer of $\N^2$.}
\end{figure}

\paragraph{Notation.} Throughout this section, we will define the sets $G_n, L_n, H_n$ as in the Introduction for our choice of $G=\N^2$. Hence,
$G_0 = \N^2$, $H_n = \CH{(G_{n-1})}$, $L_n$ is the set of vertices of
$H_n$, and $G_n = G_{n-1} \setminus L_n$. Let $B_n = \partial
H_n$. Let $(K_n, K_n)$ be the point of intersection of $B_n$ with the
line $y=x$, so the point $(K_n, K_n)$ splits $B_n$ into two congruent
``arms''. Let $s(n) = |L_1| +\cdots+|L_n|$ be the number of grid points removed
up to iteration $n$. Given integers $x,n\in \N$, let $a_x(n)$ be the
number of points of the $x$\th column of $G$ that have been removed up
to iteration $n$; i.e., let
$a_x(n) = \bigl|\bigl\{(\{x\} \times \N) \cap (L_1 \cup \cdots \cup
L_n)\bigr\}\bigr|$. Note that in a fixed column, the points are
removed in increasing order of $y$-coordinate; furthermore, for every
fixed $n$ the sequence $a_0(n), a_1(n), a_2(n), \ldots$ is
nonincreasing, with $a_0(n) = n$, $a_{n-1}(n) = 1$, and $a_n(n) = 0$. See \autoref{fig_N2layers}.

\subsection{Experiments}

Given $n$, let $B'_n$ be the result of scaling down the $n$th layer boundary $B_n$ by a factor of $n^{3/4}$. According to \autoref{conj:A:C:S:F}, as $n\to \infty$ we would expect $B'_n$ to converge to the hyperbola $\gamma_L(1/c)$ of (\ref{eq_gammaL}). We would like to measure to what extent this happens. However, since we do not know the constant $c$ to high precision, we performed these measurements as follows:

Given $n$, let $f_n(x) = K_n^2/x$ define the hyperbola that passes through the point $(K_n,K_n)$. Given a small real number $0<\alpha<1$, let $x_\alpha(n)$ denote the smallest integer $x>K_n$ for which $|a_x(n) - f(x)| > \alpha f(x)$. Hence, the portion of $B_n$ that is between $x$-coordinates $x=K_n$ and $x=x_\alpha(n)$ is within an $\alpha$-fraction of the hyperbola $f(x)$. By symmetry, the same can be said about the portion of $B_n$ that is between $y$-coordinates $y=K_n$ and $y=x_\alpha(n)$. The ratio $x_\alpha(n) / K_n$ provides a scale-independent measure of the extent to which $B_n$ is $\alpha$-close to a hyperbola.

\begin{figure}
\centerline{\includegraphics{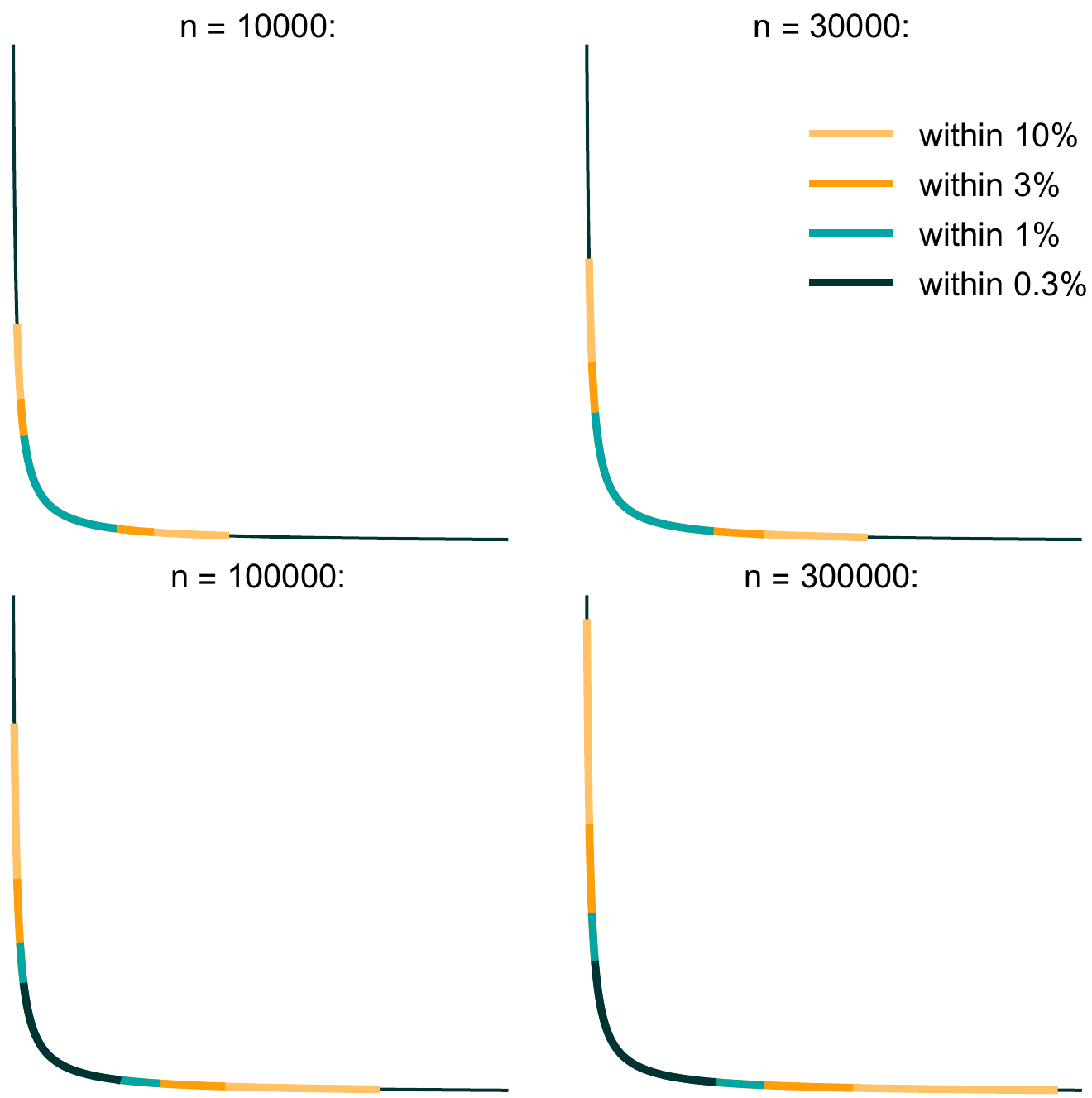}}
\caption{\label{fig_hypMeasure}Measurements indicating that as $n$ increases, an increasingly long portion of $B_n$ is close to a hyperbola.}
\end{figure}

\autoref{fig_hypMeasure} illustrates the results of these measurements for increasing values of $n$ and decreasing values of $\alpha$. Specifically, we took $n = 10000, 30000, 100000, 300000$ and $\alpha = 0.1, 0.03, 0.01, 0.003$. As can be seen, for each fixed $\alpha$, the portion of the hyperbola that is within an $\alpha$-fraction of $B_n$ grows as $n$ increases.

\begin{figure}
\centerline{\includegraphics{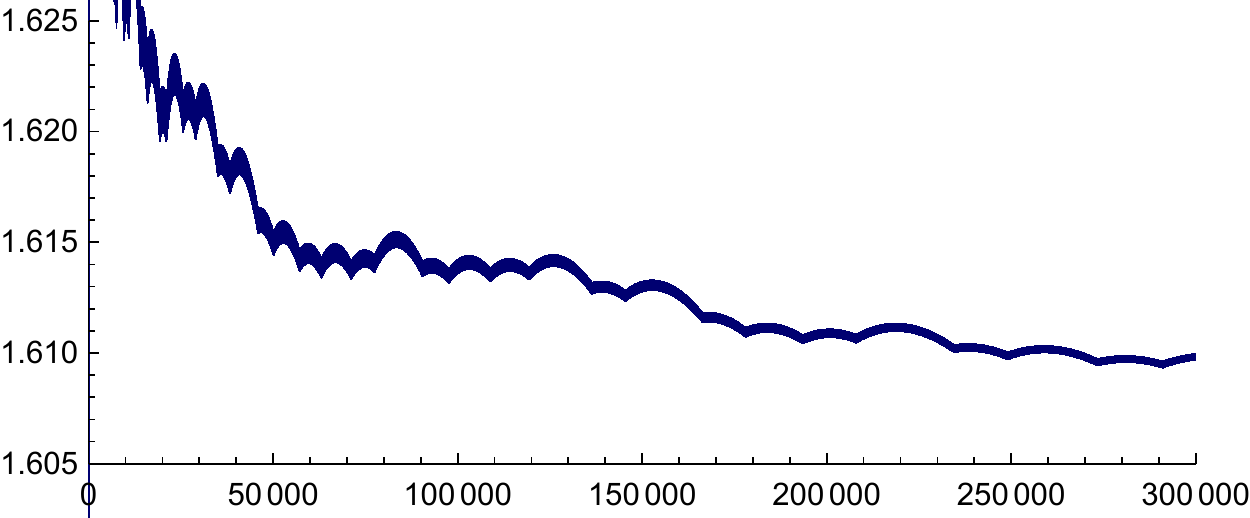}}
\caption{\label{fig_N2c}Approximations of the constant $c$ given by experimental measurements of the point $(K_n, K_n)$.}
\end{figure}

\paragraph{Measuring the time relationship.} \autoref{conj:A:C:S:F} predicts that $K_n \approx 2(n/(3c))^{3/4}$ for large $n$. We tested this prediction by solving for $c$ and seeing whether it approaches a constant close to $1.6$. The results are shown in \autoref{fig_N2c}.

\paragraph{Layer sizes.} The sequence $\{|L_n|\}$ represents the number of vertices in successive layers of the convex-layer decomposition of $\N^2$. This sequence is now cataloged in the Online Encyclopedia of Integer Sequences as A293596. It starts as
\begin{equation*}
1, 2, 2, 3, 4, 4, 3, 4, 6, 6, 5, 4, 6, 6, 8, 7, 6, 6, 6, 8, 9, 10, 10, 8, 8, 7, 8, 10, 10, 12, 13, 12, \ldots
\end{equation*}

\begin{figure}
    \centerline{\includegraphics{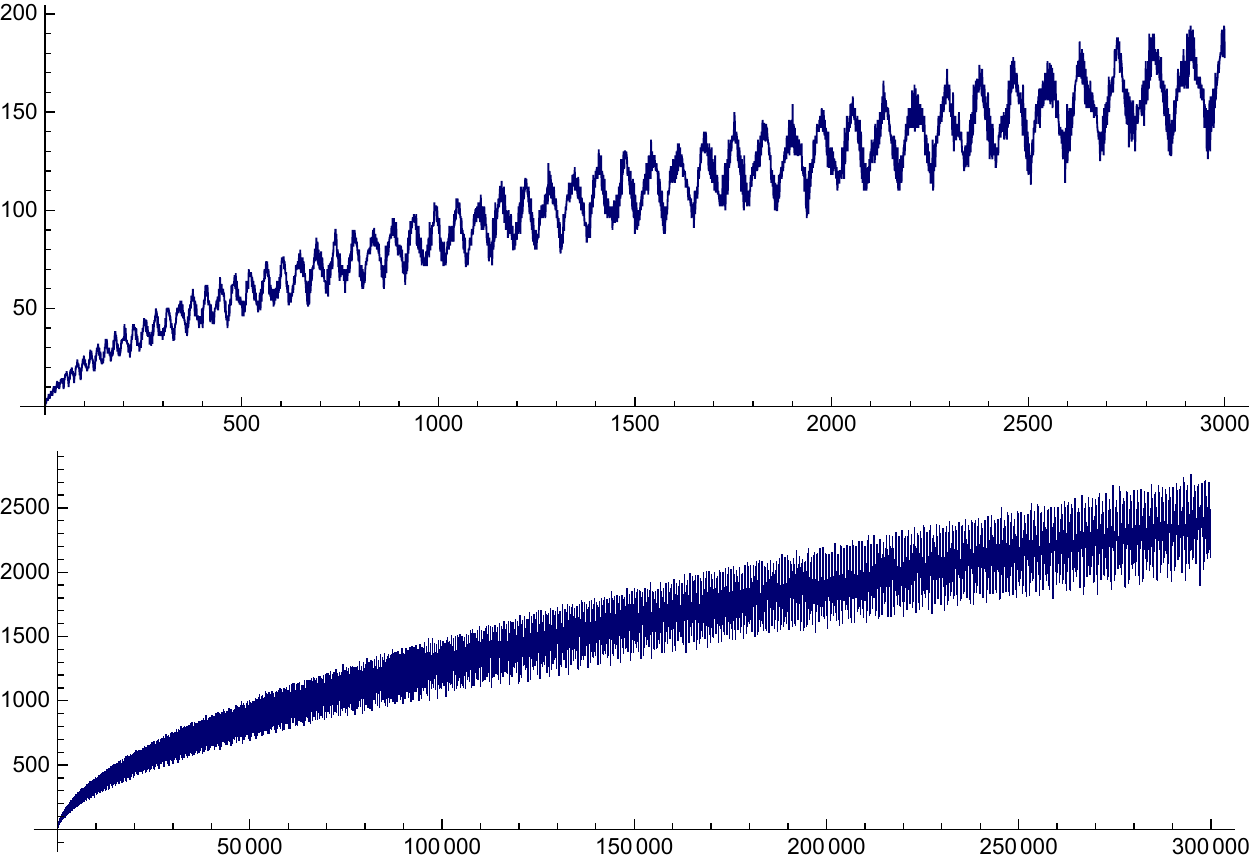}}
    \caption{\label{fig_N2:layer:sizes}%
       Number of vertices of the $n$\th
       convex layer of $\N^2$ as a function of $n$.}
\end{figure}

\autoref{fig_N2:layer:sizes} plots this sequence at different scales. As can be seen, this sequence has regular waves that slowly increase in both length and
amplitude. Nevertheless, the relative amplitude of the waves (the ratio between their amplitude and their height above the $x$-axis) seems to decrease, perhaps tending to zero.

\subsection{Rigorous results}

We now obtain some rigorous results for the grid peeling of $G=\N^2$.

\begin{figure}
    \centerline{\includegraphics{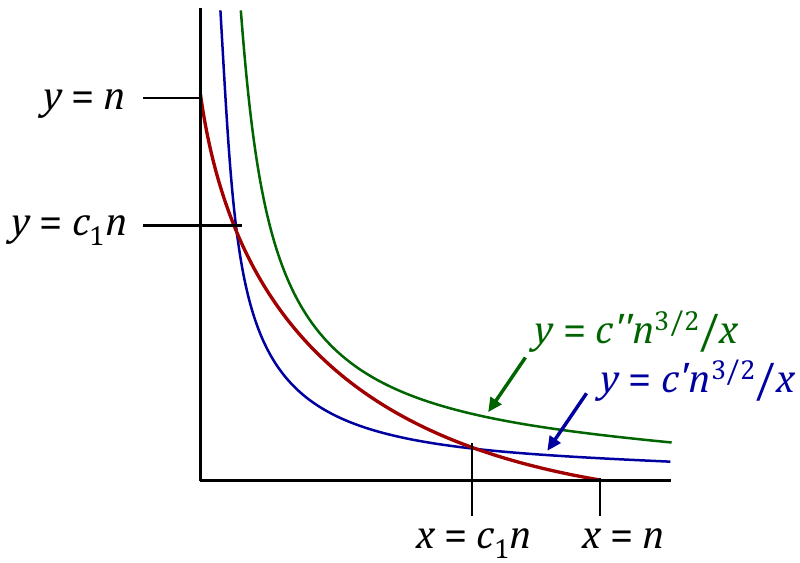}}
    \caption{\label{fig_2_hyperbolas}Schematic drawing of the upper
       and lower bounds for the $n$\th convex layer of $\N^2$.}
\end{figure}

\begin{theorem}\label{theorem_peel_corner}
    The convex-layer decomposition of the quarter-infinite grid
    $G=\N^2$ satisfies the following properties:
    \begin{enumerate}
        \item The number of grid points removed up to iteration $n$ satisfies $s(n) = \Theta(n^{3/2}\log n)$.
        \item $a_x(n) = O(n^{3/2}/x)$ for all $n$, and
        $a_x(n)=\Omega(n^{3/2}/x)$ for $c_1\sqrt n\le x \le c_1 n$,
        where $0<c_1<1$ is some constant.
        \item $K_n = \Theta(n^{3/4})$.
        \item $|L_n| = O(n^{1/2}\log n)$ and $|L_n| = \Omega(\log n)$.
    \end{enumerate}
\end{theorem}

In other words, the boundary $B_n$ is sandwiched between two hyperbolas that are separated from each other by a constant factor, where the  upper hyperbola bounds $B_n$ for all $x$, while the lower hyperbola bounds $B_n$ only up to $x = c_1 n$ (and symmetrically in the $y$-axis). Put differently, the scaled boundary $B'_n$ is sandwiched between two hyperbolas $y \ge c_1 x$ and $y\le c_2 x$ that are independent of $n$, where the lower hyperbola bounds $B'_n$ only up to $x = c_1 n^{1/4}$ (and symmetrically in the $y$-axis). See \autoref{fig_2_hyperbolas}.

Regarding $|L_n|$, we would expect it to behave like $\Theta(n^{1/2} \log n)$, or even like $c' n^{1/2} \log n \pm o(n^{1/2} \log n)$ for some constant $c'$. However, our rigorous lower bound for $|L_n|$ is very weak.

\subsection{Proof of \autoref{theorem_peel_corner}}

The proof of \autoref{theorem_peel_corner} is mainly based on the
techniques of~\cite{hl-pg-13}.

\begin{lemma}[\Jarnik~\cite{j-ugkk-26}]%
    \label{lemma_convex_in_grid}%
    Let $P\subset \{1,\ldots,m\}\times\{1,\ldots,n\}$ be in convex
    position. Then $|P| = O((mn)^{1/3})$.
\end{lemma}

\begin{proof}
    Let $p_0, p_1, p_2 \ldots, p_{k-1}$ be the points of $P$ listed in
    circular order around the boundary of $\CH{(P)}$, and let
    $v_i = p_{(i+1)\bmod k} - p_i$ be the vectors corresponding to the
    edges of $\CH{(P)}$. Note that these vectors are pairwise
    distinct. Let $\xmax = m^{2/3}n^{-1/3}$ and
    $\ymax = n^{2/3}m^{-1/3}$. Classify the vectors $v_i=(x_i,y_i)$
    into three types as follows: (1) Those satisfying
    $|x_i| \le \xmax$ and $|y_i| \le \ymax$; (2) those satisfying
    $|x_i| > \xmax$; (3) the remaining ones (which satisfy
    $|y_i| > \ymax$). The number of vectors of type (1) is at most
    $4\xmax\ymax = O((mn)^{1/3})$. The number of vectors of type (2)
    is at most $2m/\xmax = O((mn)^{1/3})$. And the number of vectors of
    type (3) is at most $2n/\ymax = O((mn)^{1/3})$.
\end{proof}

A vector $v=(x,y)\in \Z^2$ is said to be \emph{primitive} if $x$ and
$y$ are relatively prime.

\begin{lemma}\label{lem_rp}
    Let
    $M = \{a+1, \ldots, a+m\} \times \{b+1, \ldots, b+n\} \subseteq
    \{1, \ldots, N\}^2$. Then the number of primitive vectors in $M$
    is $(6/\pi^2)mn \pm O(N\log N)$.
\end{lemma}

\begin{proof}
We start with the following classical number-theoretical result.

\begin{lemma}\label{lemma_r:p_anchored}
    Let $m,n$ be positive integers with $m\le n$. Let $\rho(m,n)$ be
    the number of primitive vectors $(x,y)$ in
    $\{1, \ldots, m\}\times\{1, \ldots, n\}$. Then
    $\rho(m,n) = (6/\pi^2)mn \pm O(n\log n)$.
\end{lemma}

\begin{proof}
    (Following Hardy and Wright~\cite{hw-itn-08}, Theorem 332.) Let
    $\mu$ be the \Mobius function, which sets $\mu(x) = -1$ if $x$ is
    square-free and has an odd number of prime factors, $\mu(x) = 1$
    if $x$ is square-free and has an even number of prime factors, and
    $\mu(x) = 0$ if $x$ is not square-free. Let $D(x,y)$ be the set of
    all common divisors of $x$ and $y$. Then
    $\sum_{d \in D(x,y)} \mu(d)$ equals $1$ if $x$ and $y$ are
    relatively prime, and $0$ otherwise.

    Clearly, $\sum_{x=1}^\infty \mu(x)/x^2$ converges to some positive
    real number smaller than $1$. In fact, it converges to
    $6/\pi^2$~\cite{hw-itn-08}.

    Therefore,
    \begin{multline*}
        \rho(m,n) = \sum_{x=1}^m\sum_{y=1}^n \sum_{d\in D(x,y)} \mu(d) = \sum_{d=1}^{mn} \mu(d) \left\lfloor \frac{m}{d} \right\rfloor \left\lfloor\frac{n}{d}\right\rfloor = \sum_{d=1}^{mn} \mu(d)\left(\frac{mn}{d^2} - O\left(\frac{n}{d}\right)\right)\\
        = mn\sum_{d=1}^{mn} \frac{\mu(d)}{d^2} \pm O(n\log n)
        =mn\left(\sum_{d=1}^\infty \frac{\mu(d)}{d^2} \pm
            O\left(\frac{1}{mn}\right)\right) \pm O(n\log n),
    \end{multline*}
    and the claim follows.
\end{proof}

Now, consider $M = \{a+1, \ldots, a+m\} \times \{b+1, \ldots, b+n\} \subseteq \{1, \ldots, N\}^2$. The number of primitive vectors in $M$ equals $\rho(a+m,b+n) - \rho(a+m,b) - \rho(a,b+n) + \rho(a,b)$, so \autoref{lem_rp} follows by \autoref{lemma_r:p_anchored}.
\end{proof}

\subsubsection{Upper bounds}

\begin{lemma}%
    \label{lem_upper_bd_a}%
    We have $a_x(n)\le c_0n^{3/2}/x$ for some constant $c_0$.
\end{lemma}

\begin{proof}
    Given $x$, let $y=a_x(n)$. By iteration $n$, the entire corner subgrid
    $G'=\{0,\ldots,x-1\}\times\{0,\ldots,y-1\}$ has been removed. By
    \autoref{lemma_convex_in_grid}, each $L_i$ contains $O((xy)^{1/3})$
    points of $G'$. Hence, we must have $xy \le O(n(xy)^{1/3})$, which
    implies $y = O(n^{3/2}/x)$.
\end{proof}

\begin{corollary}
    We have $s(n) = O(n^{3/2} \log n)$.
\end{corollary}

\begin{corollary}\label{cor_Kn}
    We have $K_n \le \sqrt{c_0}n^{3/4}$ for the constant $c_0$ of
    \autoref{lem_upper_bd_a}.
\end{corollary}

\begin{proof}
    Take $x_0 = \sqrt{c_0} n^{3/4}$, and note that
    $a_{x_0}(n) \le x_0$.
\end{proof}

By \autoref{cor_Kn}, each ``arm'' of $L_n$ is contained in an
$O(n^{3/4})\times n$ box. Hence, a hasty application of
\autoref{lemma_convex_in_grid} would yield $|L_n| =
O(n^{7/12})$. However, we can do better: We can cover each arm of
$L_n$ by logarithmically many boxes of small area, and apply
\autoref{lemma_convex_in_grid} on each box.

\begin{lemma}
    We have $|L_n| = O(n^{1/2}\log n)$.
\end{lemma}

\begin{proof}
    Let $x_0 = \sqrt{c_0} n^{3/4}$ for the constant $c_0$ of
    \autoref{lem_upper_bd_a}, and recall that $K_n \le x_0$. Define
    the axis-parallel boxes $T_i$,
    $0\le i\le \lceil\log_2 (n/x_0)\rceil$, by $T_{0} = [0,x_0]^2$ and
    $T_i = [2^{i-1}x_0,2^i x_0]\times [0,x_0 / 2^i]$ for $i\ge 1$. By
    \autoref{lem_upper_bd_a}, the right arm of $L_n$ is contained in
    the union of these boxes. Furthermore, the area of $T_{0}$ is
    $c_0 n^{3/2}$, and the area of each $T_i$, $i\ge 1$, is
    $c_0 n^{3/2}/2$. Hence, by \autoref{lemma_convex_in_grid}, each
    $T_i$ contains $O(\sqrt{n})$ points of $L_n$. Finally, the number
    of boxes is $O(\log n)$.
\end{proof}

\subsubsection{Lower bounds}

\begin{figure}
    \centerline{\includegraphics{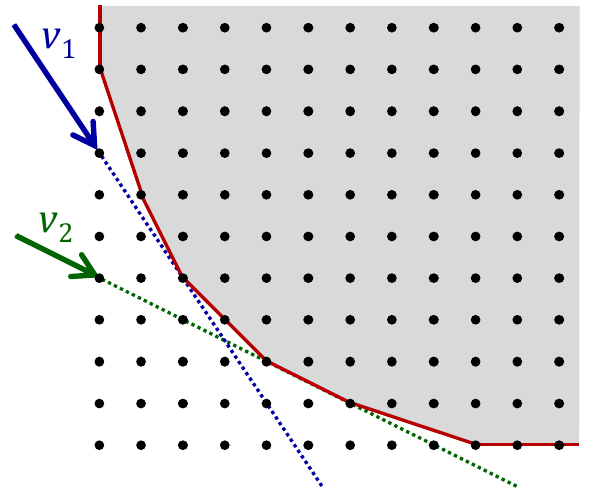}}
    \caption{\label{fig_active}At the shown iteration, vector
       $v_1=(2,-3)$ is inactive, while vector $v_2=(2,-1)$ is active.}
\end{figure}

Let $v=(x_v,-y_v)$ be a primitive vector with $x_v,
y_v>0$. Following~\cite{hl-pg-13}, we say that $v$ is \emph{active} at
iteration $n$ if the unique line $\ell_v$ parallel to $v$ that is
tangent to $H_n$ contains an edge of $H_n$ (and so $\ell_v$ contains
two points of $L_n$). Otherwise, if $\ell_v$ contains a single vertex
of $H_n$ (and a single point of $L_n$), then we say that $v$ is
\emph{inactive} at iteration $n$. See \autoref{fig_active}.

Given such a vector $v$, let $\mathcal L_v$ be the set of lines
parallel to $v$ that pass through points of $\N^2$. We say that a line
$\ell \in \mathcal L_v$ is \emph{alive} at iteration $n$ if $\ell$
intersects $H_n$; otherwise, we say that $\ell$ is \emph{dead} at
iteration $n$. Note that, at a given iteration $n$, all the dead lines
of $\mathcal L_v$ lie below all its live lines.

\begin{lemma}\label{lem_Lv_rect}
    Let $v=(x_v,-y_v)$ be a primitive vector with $x_v, y_v>0$. Then
    the number of lines of $\mathcal L_v$ that pass below
    the point $(x,y)$ is at most $xy_v + yx_v$.
\end{lemma}

\begin{proof}
    Each of the said lines passes through a grid point in $\{0, \ldots, x_v-1\}\times\{0, \ldots, y + \lceil x y_v/x_v\rceil\}$.
\end{proof}

\begin{observation}\label{obs_line_dies}
    If $v$ is inactive at iteration $n$, then the number of dead lines
    of $\mathcal L_v$ strictly increases from iteration $n$ to
    iteration $n+1$; specifically, the tangent line $\ell_v$ dies.
\end{observation}

\begin{observation}\label{obs_num_active}
    The number of active vectors at iteration $n$ equals $|L_n|-1$.
\end{observation}

Given $n$, let $m = n^{1/2}/(16c_0)$ for the constant $c_0$ in \autoref{lem_upper_bd_a}. Let $\mathcal V$ be the set of all primitive vectors $(x_v, -y_v)$ with $x_v,y_v>0$ and $x_vy_v \le m$. By applying \autoref{lem_rp} on the rectangles $\{2^{i-1}\sqrt m,\ldots,\allowbreak 2^i\sqrt m\}\times\{0,\ldots,2^{-i}\sqrt m\}$ for $-\log_5 m\le i\le \log_5 m$, we obtain $|\mathcal V| = \Theta(m\log m) = \Theta(n^{1/2}\log n)$.

\begin{lemma}
Up to iteration $n$, each vector of $\mathcal V$ is active at least $n/2$ times.
\end{lemma}

\begin{proof}
Consider a vector $v=(x_v, -y_v)\in\mathcal V$. Set $x = c_0^{1/2}n^{3/4}\sqrt{x_v/y_v}$ and $y=1+a_x(n)\le 1+ c_0^{1/2}n^{3/4}\sqrt{y_v/x_v}$ (by \autoref{lem_upper_bd_a}). Since the grid point $(x,y)$ has not been removed by iteration $n$, by \autoref{lem_Lv_rect} the number of dead lines parallel to $v$ is at most
\begin{equation*}
xy_v+yx_v = 2 c_0^{1/2}n^{3/4} \sqrt{x_vy_v} \le 2 c_0^{1/2}n^{3/4} \sqrt{m} = n/2.
\end{equation*}
Hence, the claim follows from \autoref{obs_line_dies}.
\end{proof}

Hence, by \autoref{obs_num_active}:

\begin{corollary}\label{cor_lower_s}
    We have $s(n) = \Omega(n^{3/2}\log n)$.
\end{corollary}

\begin{lemma}
    There exist constants $c_1,c_2>0$ such that, for every $n$ and
    every $x$ in the range $c_1\sqrt{n}\le x\le c_1 n$, we have
    $a_x(n) \ge c_2 n^{3/2}/x$.
\end{lemma}

\begin{proof}
    Given a slope $\mu$ in the range $1/m\le \mu\le m$, we will derive
    a lower bound for the distance between the origin and $B_n$ in the
    direction $\mu$. (So for example, taking $\mu=1$ will yield a
    lower bound for $K_n$.)

    Using \autoref{lem_upper_bd_a}, take a grid point $(x,y)$ with
    $y/x\approx \mu$ that has not been removed by iteration
    $n$. Specifically, let $x=\bigl(\sqrt{c_0/\mu}\bigr)n^{3/4}$ and
    $y=\bigl(\sqrt{\mu c_0}\bigr)n^{3/4}$. Define the rectangle
    $T=\{0,\ldots,3x-1\}\times\{0,\ldots,3y-1\}$, so
    $|T|=9c_0n^{3/2}$. We claim that at least a constant fraction of
    the points of $T$ have been removed by iteration $n$.

    Indeed, let $\mathcal V'\subseteq \mathcal V$ be the set of all
    vectors $v=(x_v,-y_v)\in\mathcal V$ with
    $\mu/2 \le y_v/x_v \le 2\mu$. By applying \autoref{lem_rp} on the
    rectangle whose opposite corners are $q/2$ and $q$ for
    $q=(\sqrt{m/\mu}, \sqrt{\mu m})$, we have
    $|\mathcal V'|=\Theta(m)=\Theta(\sqrt n)$.

    Let $v\in\mathcal V'$, and let $i\le n$ be an iteration in which
    $v$ is active. Let $\ell\in\mathcal L_v$ be the line tangent to
    $H_i$. Line $\ell$ passes below point $(x,y)$, so by the
    construction of $T$, all the grid points in $\ell$ belong to
    $T$. Two of these grid points belong to $L_i$. Let us charge the
    pair $(v,i)$ to the leftmost of these two points.

    Doing this over all choices of $v$ and $i$, we make a total of
    $\Theta(n^{3/2})$ charges to points of $T$. Furthermore, each point of $T$
    is charged at most once. Therefore, at least a constant
    fraction (say, a $c'$-fraction) of the points of $T$ are deleted
    by iteration $n$.

    Choose a constant $0<c'' < 1-\sqrt{1-c'}$. Let $x'=3c''x$ and
    $y'=3c''y$ (so $y'/x'=\mu$). We claim that the grid point
    $(x',y')$ has been removed by iteration $n$. Indeed, otherwise,
    all the points behind $(x',y')$ (i.e. all the points
    $(x'',y'')\in T$ with $x''\ge x'$ and $y''>y'$) would also be
    present, and they constitute more than a $(1-c')$-fraction of $T$
    (by the choice of $c''$).

    Rephrasing, given $n$ and given $x'$ in the range
    ${3c''\sqrt{c_0}\sqrt{n}}\le x'\le {3c''\sqrt{c_0}n}$, we have
    \begin{equation*}
        a_{x'}(n) \ge y' = 3c''y=3c''c_0n^{3/2}/x=9(c'')^2c_0n^{3/2}/x'.
    \end{equation*}
\end{proof}

\begin{corollary}
    We have $|L_n| = \Omega(\log n)$.
\end{corollary}

\begin{proof}
The idea is that, since $B_n$ is confined between two hyperbolas for a long stretch, it must make at least a certain number of turns. That number is $\Omega(\log n)$, by the following calculation:

For simplicity, let us scale down $B_n$ by a factor of $n^{3/4}$, obtaining $B'_n$. Let $y=c_1/x$ and $y=c_2/x$ be the two bounding hyperbolas, where $c_1<c_2$, and where the lower-hyperbola bound applies up to $x = c_3 n^{1/4}$. The number of edges of $B'_n$ is minimized if each edge starts and ends at the lower hyperbola and is tangent to the upper hyperbola. In such a case, an edge that starts at $x$-coordinate $x_0$ ends at $x$-coordinate $k x_0$, for the constant $k = (\sqrt{c_2} + \sqrt{c_2-c_1})/(\sqrt{c_2}-\sqrt{c_2-c_1})$.

Hence, the number of edges is at least $\log_k n^{1/4} - O(1)$.
\end{proof}

\section{Concluding remarks}

The main open problem is to prove
\autoref{conj:A:C:S:F}. Additionally, if the conjecture can be
confirmed, it would be of interest to generalize the approximation to
the \ACSF that it yields, from convex curves to more general curves.
Also, grid peeling for higher dimensions has not been studied at all,
as far as we know.

\paragraph{Acknowledgements.} The third author would like to thank Franck Assous and Elad Horev for useful conversations.

\bibliographystyle{amsplainurl}
\bibliography{peeling}

\end{document}